\documentclass[11pt]{article}

\newif\ifnotes
\notestrue

\usepackage{hyperref}
\usepackage{fullpage}
\usepackage{amssymb}
\usepackage{amsmath}
\usepackage{mathtools}

\usepackage{amsthm}
\usepackage{amsfonts}
\usepackage{bm}
\usepackage{enumitem}
\usepackage{color}
\usepackage{comment}
\usepackage[capitalize]{cleveref}
\usepackage[dvipsnames]{xcolor}
\usepackage{float}
\usepackage[T1]{fontenc}
\usepackage{todonotes}
\usepackage{asymptote}
\usepackage{mdframed}
\usepackage[most]{tcolorbox}
\usepackage{hyperref}
\usepackage{enumitem}
\usepackage{framed}
\usepackage{mdframed}
\usepackage{scrextend}
\usepackage{bbm}

\usepackage[T1]{fontenc}

\definecolor{asparagus}{rgb}{0.53, 0.66, 0.42}
\definecolor{sacramentostategreen}{rgb}{0.0, 0.3, 0.15}
\definecolor{teal}{rgb}{0.0, 0.5, 0.5}
\definecolor{forestgreen}{rgb}{0.13, 0.6, 0.13}

\usepackage{hyperref}
\hypersetup{
    colorlinks=true,
    linkcolor=sacramentostategreen,
    filecolor=sacramentostategreen,
    citecolor=teal,
    urlcolor=teal,
}

\newtheorem{theorem}{Theorem}[section]
\newtheorem{definition}[theorem]{Definition}
\newtheorem{lemma}[theorem]{Lemma}

\newtheorem{corollary}[theorem]{Corollary}

\newtheorem*{problem*}{Problem}

\theoremstyle{remark}

\Crefname{theorem}{Theorem}{Theorems}
\Crefname{claim}{Claim}{Claims}
\Crefname{lemma}{Lemma}{Lemmas}
\Crefname{proposition}{Proposition}{Propositions}
\Crefname{corollary}{Corollary}{Corollaries}
\Crefname{definition}{Definition}{Definitions}

\newcommand{\srch}{\mathsf{search}}
\newcommand{\error}{\mathsf{ERROR}}
\newcommand{\bcert}{\mathsf{binary\_cert}}

\newcommand{\R}{\mathbb{R}}
\newcommand{\Z}{\mathbb{Z}}
\newcommand{\D}{\mathbb{D}}
\newcommand{\cA}{\mathcal{A}}

\newcommand{\cP}{\mathcal{P}}

\newcommand{\defeq}{\coloneqq}
\newcommand{\onev}{\mathbbm{1}}

\newcommand{\suchthat}{{\;\; : \;\;}}

\newcommand{\inparen}[1]{\left(#1\right)}             
\newcommand{\inbraces}[1]{\left\{#1\right\}}           
\newcommand{\insquare}[1]{\left[#1\right]}             

\newcommand{\B}{\inbraces{0,1}}

\newcommand{\indicator}[1]{\mathbbm{1}\inbraces{#1}}

\newcommand{\abs}[1]{\ensuremath{\left\lvert #1 \right\rvert}}

\newcommand{\myfunc}[3]{#1\colon#2\to#3}

\usepackage{nicefrac}

\let\nfrac=\nicefrac

\renewcommand{\Pr}{\mathsf{Pr}}

\newcommand{\prvv}[2]{\underset{#1}{\Pr}\insquare{#2}}

\makeatletter
\newcommand{\customlabel}[2]{%
   \protected@write \@auxout {}{\string \newlabel {#1}{{#2}{\thepage}{#2}{#1}{}} }%
   \hypertarget{#1}{#2}
}
\makeatother

\newcommand{\protocol}[3]{
    \stepcounter{figure}
    \vspace{0.15cm}
    { \small
    \begin{tcolorbox}[breakable, enhanced, colback=forestgreen!12]
    \begin{center}
    {\bf \underline{Algorithm~\customlabel{prot:#2}{\thefigure}: #1}}
    \end{center}
    
    #3
    \end{tcolorbox}
    }
}

\newcounter{casenum}

\newcounter{subcasenum}

\newcounter{casenump}

\newcommand{\casep}[2]{
    \ifthenelse{\equal{\value{casenump}}{0}}{
    \vskip.5\baselineskip\par\noindent
    }{}
    {\it Case \arabic{casenump}:} {\it #1}
    \vskip0.1\baselineskip
    \begin{addmargin}[1.5em]{1em}
    #2
    \end{addmargin}
    \addtocounter{casenump}{1}
}

\newcounter{subcasenump}

\usepackage[
backend=biber,
style=trad-alpha,
maxalphanames=8,
sorting=nty
]{biblatex}
\newcommand{\xstar}{x^{\star}}

\begin{document}

\title{An Optimal Algorithm for Certifying Monotone Functions}  
\author{Meghal Gupta \thanks{E-mail:\texttt{meghal@mit.edu}}\\Microsoft Research \and Naren Sarayu Manoj\thanks{E-mail:\texttt{nsm@ttic.edu}.}\\Toyota Technological Institute Chicago}
\date{\today}

\sloppy
\maketitle
\begin{abstract}
Given query access to a monotone function $\myfunc{f}{\B^n}{\B}$ with certificate complexity $C(f)$ and an input $\xstar$, we design an algorithm that outputs a size-$C(f)$ subset of $\xstar$ certifying the value of $f(\xstar)$. Our algorithm makes $O\inparen{C(f) \cdot \log n}$ queries to $f$, which matches the information-theoretic lower bound for this problem and resolves the concrete open question posed in the STOC '22 paper of Blanc, Koch, Lange, and Tan \cite{BKLT22}.

We extend this result to an algorithm that finds a size-$2C(f)$ certificate for a real-valued monotone function with $O\inparen{C(f) \cdot \log n}$ queries. We also complement our algorithms with a hardness result, in which we show that finding the shortest possible certificate in $\xstar$ may require $\Omega\left(\binom{n}{C(f)}\right)$ queries in the worst case. 
\end{abstract}
\thispagestyle{empty}

\section{Introduction}
Given a function $\myfunc{f}{\B^n}{\D}$ for some output domain $\D$ and an input $\xstar$, is there a short proof for why $f(\xstar)$ takes on the value it does? This natural question motivates the notion of \emph{certificate complexity} in complexity theory. Loosely speaking, a certificate for $f(\xstar) = y$ is a subset of the bits of $\xstar$ that ``fixes'' the value of $f(\xstar)$. In other words, every input $x$ that agrees with $\xstar$ on the bits in the certificate will satisfy $f(x) = f(\xstar)$. Besides being a quantity of interest in complexity theory and in the analysis of Boolean functions, certificate complexity has a natural interpretation in the context of explainable AI. Here, the practitioner aims to find simple properties of a given input that explain a classifier's prediction on the input.
We formalize the notion of a certificate in Definition \ref{def:certificate}.

\begin{definition}[Certificate (see, e.g., \cite{arorabarak})]
\label{def:certificate}
Let $x\vert_S$ denote the substring of $x$ in the coordinates of $S$.

For a function $f:\B^n\to \D$ and an input $\xstar\in\B^n$, we say a set $S \subseteq [n]$ is a certificate if for all $y\in\B^n$ such that $\xstar\vert_S=y\vert_S$, we have $f(\xstar)=f(y)$.
\end{definition}

We use Definition \ref{def:certificate} to define the certificate complexity of a function $f$.

\begin{definition}[Certificate complexity (see, e.g., \cite{arorabarak})]
\label{def:cert_complexity}
For any function $\myfunc{f}{\B^n}{\D}$ and $x \in \B^n$, we let $C(f,x)$ be the smallest integer such that there exists a $C(f,x)$-sized certificate for $f(x)=j$. We now let the certificate complexity of $f$ be $\max\limits_{x\in\B^n} C(f,x)$.
\end{definition}

A natural follow-up question from Definition \ref{def:cert_complexity} is whether a short certificate can be found in a given input if we know that all inputs have a short certificate. The following problem, posed and studied in the STOC '22 paper of \cite{BKLT22}, formalizes this question.

\begin{problem*}
    Given queries to a function $f : \B^n \to \B$ with certificate complexity $C(f)$ and an input $\xstar$, output a size-$C(f)$ certificate for $f$’s value on $\xstar$.
\end{problem*}

The main result of \cite{BKLT22} is an algorithm for the case where $f$ is monotone and the output range $\D=\B$. The authors design a randomized algorithm that makes at most $O\inparen{C(f)^8\cdot\log{n}}$ queries using a novel connection to threshold phenomena. Furthermore, \cite{BKLT22} show that $\Omega(C(f)\cdot\log{n})$ queries for the certification problem are necessary in the worst-case. The authors identify closing this gap as a concrete direction for future work.

\subsection{Our Results}

Our main result is a simple, deterministic algorithm that makes $O(C(f)\cdot\log{n})$ queries to find a size-$C(f)$ certificate for any monotone binary-valued function $f$ and input $\xstar$. This completely resolves the aforementioned open question from \cite{BKLT22}. Formally, we have Theorem \ref{thm:main_result_intro}.

\begin{theorem}
\label{thm:main_result_intro}
Given query access to a monotone function $\myfunc{f}{\B^n}{\B}$ and an input $\xstar$, there exists an algorithm that makes $O(C(f)\cdot\log n)$ queries to $f$ and outputs a size-$C(f)$ subset $S$ corresponding to a subset of indices of $\xstar$ certifying the value of $f(\xstar)$.
\end{theorem}

We can extend our result to obtain as a simple corollary an algorithm that finds a size-$2C(f)$ certificate for any monotone \textbf{real}-valued function $f$ and input $\xstar$. Specifically, we have Theorem \ref{thm:main_result_reals_intro}.

\begin{theorem}
\label{thm:main_result_reals_intro}
Given query access to a monotone function $\myfunc{f}{\B^n}{\R}$ and an input $\xstar$,there exists an algorithm that makes $O(C(f)\cdot \log n)$ queries to $f$ and outputs a size-$2C(f)$ subset $S$ corresponding to a subset of indices of $\xstar$ certifying the value of $f(\xstar)$.
\end{theorem}

The careful reader might also wonder why we are only looking for a certificate of size-$C(f)$ on the input $\xstar$ -- by definition, the shortest certificate on a fixed input $\xstar$ is size-$C(f,\xstar)$. We show that finding a certificate of length $C(f,\xstar)$ may require far more queries than simply finding one of length $C(f)$. In particular, in the case where $C(f,\xstar) = \nfrac{n}{2}$, it may require exponentially many queries. Moreover, our result matches the trivial upper bound provided by an algorithm which simply queries all $C(f,\xstar)$-size certificates. See Theorem \ref{thm:lower_bound_intro}.

\begin{theorem}
\label{thm:lower_bound_intro}
For any $k$, for any (randomized) algorithm that queries a given function, there exists a function $f$ and input $\xstar$ such that $k=C(f,\xstar)$, and the algorithm must make at least $\frac12\binom{n}{k}$ queries to determine a certificate with probability $>\nfrac{1}{2}$.
\end{theorem}

\subsection{Related Work}

We derive our setting and problem statements from the work of \cite{BKLT22}. The authors of \cite{BKLT22} formally propose the problem of certifying a monotone function $f$ on an input $\xstar$ and provide an algorithm for doing so, as mentioned earlier. They also look at the certification question for a general (non-monotone) function $f$. Here, they show $\Omega\inparen{2^{C(f)}+C(f)\cdot\log{n}}$ queries are necessary, and $O\inparen{2^{C(f)}\cdot C(f)\cdot\log{n}}$ queries suffice with high probability. Closing this gap remains an interesting open direction.

Before the work of \cite{BKLT22}, Angluin (see \cite{angluin}) gave a local search algorithm that can be used to certify a monotone function $f$ on an input $\xstar$ with query complexity $O(n)$. See Appendix C in \cite{BKLT22} for a detailed exposition and proof of correctness of Angluin's algorithm.

\section{Preliminaries}

\paragraph{Notation} In this work, we use the following notation.
\begin{itemize}
\item We denote the set $\inbraces{x \in \Z_{\ge 0} \suchthat 1 \le x \le n}$ as $[n]$. In an abuse of notation, let $[0] = \emptyset$.
\item For a set $S \subseteq [n]$, we write $x_S$ to be the indicator vector for $S$; i.e., $x_S$ is such that $x_i = \indicator{i \in S}$, for all $i \in [n]$. Additionally, we write $x\vert_S$ to be the substring of $x$ in the coordinates of $S$. Specifically, we have $x\vert_S = \inbraces{(i, x_i) \text{ for all } i \in S}$.
\item Let $\onev^n$ denote the all-$1$s vector in $n$ dimensions.
\item For a vector $x\in \{0,1\}^n$, we denote $S_{x}$ to be $\inbraces{i \suchthat x_i=1}$.
\end{itemize}

In our work, it is helpful to distinguish a \emph{minimal} certificate from a general certificate.

\begin{definition}[Minimal Certificate]
\label{def:minimal_certificate}
For a given function $f$, we say a certificate $S \subseteq [n]$ is \emph{minimal} if for all $a\in S$, we have that $S \setminus a$ is not a certificate for $f(x)$. 

If $f$ is monotone, this is equivalent to requiring that for all $A\subset S$, we have $f(x\vert_A)\neq f(x\vert_S)$.
\end{definition}

Finally, we note the information-theoretic lower bound from \cite{BKLT22} on the query complexity of any algorithm used to certify $f(x)$ for a monotone function $f$.

\begin{lemma}[Claim 1.2 in \cite{BKLT22}]
For any $c < 1$ and any $k\le l \le n^c$, let $\cA$ be an algorithm which, given query access to a monotone function $\myfunc{f}{\B^n}{\B}$ with certificate complexity $\le k$ and an input $\xstar$, returns a size-$l$ certificate for $f$'s value on $\xstar$ with high probability. The query complexity of $\cA$ must be $\Omega\inparen{k\log n}$.  
\end{lemma}
\section{Our Algorithm to Certify a Binary Monotone Function}\label{sec:bin-search}

We first restate the problem.

\begin{problem*}
Given query access to a monotone function $\myfunc{f}{\B^n}{\B}$ with certificate complexity $C(f)$ and an input $\xstar$, output a size-$C(f)$ certificate for $f$’s value on $\xstar$.
\end{problem*}

\subsection{Overview of Our Algorithm}

We informally describe our algorithm. Without loss of generality, we let $\xstar$ be such that $f(\xstar)=1$. A valid certificate is any subset $S\subset S_{\xstar}$ of indices such that $f(x_S)=1$.

We add elements into our certificate $A$ one-by-one. To do this, we simply iterate the following steps until $A$ is a valid certificate:
\begin{enumerate}
    \item Find the smallest $s\in S_{\xstar}$ such that including all $i\leq s \in S_{\xstar}$ in the certificate, along with elements already in $A$, yields a valid certificate.
    \item Add $s$ to $A$.
\end{enumerate}

Because the function is monotone, $s$ can be found through binary search at each step. Moreover, observe that removing any one element from $A$ no longer yields a valid certificate; thus, as we will show in Lemma~\ref{lemma:cf_minimal_certificate}, the output certificate is length at most $C(f)$. This also implies the algorithm makes a total of $O(C(f)\cdot\log{n})$ queries.

\subsection{Formal Description of Our Algorithm}\label{algo}

We state our algorithm formally. In our algorithm description and analysis, we assume without loss of generality that $f(\xstar)=1$. We can make this assumption since if $f(\xstar)=0$, we can instead run the algorithm making queries to $g(x) \defeq 1-f(\onev^n-x)$, which is a monotone function with $g(\xstar)=1$.

\begin{definition}[$\srch$]
\label{def:search}
The procedure $\srch(f,A,S)$ acts on a monotone function $f\in \inbraces{0,1}^n\to \inbraces{0,1}$, and two sets $A,S\subseteq [n]$. If $f(x_A)=1$ or $f(x_{A\cup S})=0$, it outputs $\error$. Else, it outputs the smallest $s\in S$ for which $f(x_{A\cup([s]\cap S)})=1$. The function proceeds using binary search, which can be done because $f$ is monotone.
\end{definition}

\protocol{Algorithm to Certify a Binary Monotone Function Where $f(\xstar) = 1$}{s33d}{
\begin{enumerate}
    \item \textbf{Input:} Query access to a function $\myfunc{f}{\inbraces{0,1}^n}{\inbraces{0,1}}$ and point $\xstar$ for which $f(\xstar) = 1$.
    \item Initialize the sets $A\gets \emptyset$ and $S\gets S_{\xstar}$
    \item\label{binary-step} Run the following procedure until $f(x_A)=1$:
    \begin{enumerate}
        \item\label{search-step} Set $s\gets \srch(f,A,S)$.
        \item Add $s$ to $A$.
        \item Set $S \gets S\cap [s-1]$
    \end{enumerate}
    \item \textbf{Output:} $A$.\label{line:output}
\end{enumerate}
}

\subsection{Analysis}

\begin{theorem}\label{mainthm}
The Algorithm in \ref{algo} outputs a certificate of length at most $C(f)$ for $f$ on $\xstar$ and makes at most $O(C(f)\cdot \log{n})$ queries.
\end{theorem}

We break the proof down into a series of lemmas.

\begin{lemma}
\label{lemma:cf_minimal_certificate}
If $S$ is a minimal certificate, then $\abs{S} \le C(f)$.
\end{lemma}
\begin{proof}
Consider the shortest certificate $C$ for the input $x_S$. We must have $f(x_C)=f(x_S)$, and $|C|\leq C(f)$. The fact that $|C|<|S|$ and $f(x_C)=f(x_S)$ contradicts that $S$ is minimal.
\end{proof}

\begin{lemma}\label{no-error}
The Algorithm in \ref{algo} never outputs $\error$.
\end{lemma}

\begin{proof}
If the algorithm outputs $\error$, it must be in Step~\ref{search-step}. By definition, an error occurs if $f(x_A)=1$ or $f(x_{A\cup S})=0$. The former cannot be true because the algorithm checks this exact condition in Step~\ref{binary-step}. The latter cannot be true because:
\begin{itemize}
    \item If this is the first iteration of Step~\ref{binary-step}, $A\cup S=S_{\xstar}$, which means $f(x_{A\cup S_{\xstar}})=f(\xstar)=1$.
    \item Else, in the previous iteration of Step~\ref{search-step} (let the values of $A,S,s$ at that step be $A', S', s'$ respectively), it must have been the case that $f(x_{A'\cup(S'\cap[s'])})=1$. Note that $A'\cup (S'\cap[s'])=A\cup(S'\cap[s'-1])=A\cup S$, and so $f(x_{A\cup S})=1$.
\end{itemize}
\end{proof}

\begin{lemma}
If the Algorithm in \ref{algo} terminates, it outputs a minimal certificate for $f$ on $\xstar$.
\end{lemma}
\begin{proof}
It must be the case that $f(A)=1$; otherwise, we could not have left Step~\ref{binary-step}. Consider any $s\in A$ and we will show that $f(A\setminus s)=0$. 

At the iteration of Step~\ref{binary-step} where $s$ was added to $A$ (let the temporary certificate $A$ at the start of that step be $A_s$), it must be the case that $f(x_{(S\cap[s])\cup A_s})=1$ but $f(x_{(S\cap[s-1])\cup A_s})=0$. All future elements that are added to create the final certificate $A$ must be a subset of $S\cap [s-1]$ (where $S$ is being referenced from the current iteration of Step~\ref{binary-step}). Therefore, $A\setminus S \subseteq (S\cap [s-1])\cup A_s$, and therefore $f(A\setminus s)=0$.
\end{proof}

\begin{lemma}
The Algorithm in \ref{algo} terminates, making at most $O(C(f)\log{n})$ queries.
\end{lemma}
\begin{proof}
Observe that in every iteration of the main loop, we add exactly one element to $A$. By Lemma \ref{lemma:cf_minimal_certificate}, there are at most $C(f)$ coordinates in the output $A$. Hence, we run the main loop at most $C(f)$ times.

Next, $\srch(f, A, S)$ is a binary search over a domain of size $\abs{S} \le \abs{S_{\xstar}} \le n$. Therefore, $\srch(f, A, S)$ uses at most $\log n$ queries.

Finally, the check $f(\xstar_A) = 1$ costs $1$ query, and this runs at the beginning of every iteration of the loop. In total, we make at most $C(f) \cdot (\log n + 1)$ queries, as desired.
\end{proof}

Combining these lemmas finishes the proof of Theorem~\ref{mainthm}.
\section{Extension to Real-Valued Functions}

In this section, we prove the following corollary of our main result wherein the output domain is $\R$ instead of $\B$.

\begin{corollary}
There exists an algorithm that, given an input $\xstar$ and query access to a monotone $\myfunc{f}{\B^n}{\R}$, makes $O(C(f) \cdot \log n)$ queries to $f$ and outputs a size-$2\cdot C(f)$ certificate for $f(\xstar)$.
\end{corollary}

\subsection{Our Algorithm}\label{reals}

We begin with two necessary definitions.

\begin{definition}[$\bcert(f,\xstar)$]
The procedure $\bcert(f,\xstar)$ is given query access to function $f:\inbraces{0,1}^n\to\inbraces{0,1}$ and an input $\xstar$, runs our algorithm from Section~\ref{sec:bin-search}, and outputs a size-$C(f)$ certificate for $f(\xstar)$.
\end{definition}

\begin{definition}[$g_{0,f,\xstar}(x)$, $g_{1,f,\xstar}(x)$]
Let $b\in \inbraces{0,1}$, $f:\inbraces{0,1}^n\to\R$ and $\xstar\in\{0,1\}^n$. The function $g_{b,f,\xstar}:\inbraces{0,1}^n\to\inbraces{0,1}$ is defined as follows:
\[
g_b(x)=\begin{cases}
0 & f(x)<f(\xstar) \\
1 & f(x)>f(\xstar) \\
b & f(x)=f(\xstar)
\end{cases}
\]
We will abbreviate $g_{b,f,\xstar}:\inbraces{0,1}^n\to\inbraces{0,1}$ as $g_b$ when $f$ and $\xstar$ are clear.
\end{definition}

\protocol{Algorithm to Certify a Real-Valued Monotone Function}{reals}{
\begin{enumerate}
    \item \textbf{Input:} $f, \xstar$.
    \item Set $C_0 \gets \bcert(g_0,\xstar)$.
    \item Set $C_1 \gets \bcert(g_1,\xstar)$.
    \item \textbf{Output:} $C_0 \cup C_1$.
\end{enumerate}
}

\subsection{Analysis}

\begin{theorem}\label{thm:reals}
The Algorithm in \ref{reals} outputs a certificate for $f$ of length at most $2C(f)$ and makes at most $O(C(f)\log{n})$ queries.
\end{theorem}

We break the proof into a series of lemmas. Call the output $A$.

\begin{lemma}
$A$ is a valid certificate for $f$ on $\xstar$.
\end{lemma}

\begin{proof}
For any input $y$ such that $y\vert_A=\xstar\vert_A$, we must have $g_b(y)=g_b(\xstar)$ for both $b=0,1$. Notice that both of the following hold:
\begin{align*}
    g_0(y)=g_0(\xstar) &\text{ implying } f(y)\leq f(\xstar) \\
    g_1(y)=g_1(\xstar) &\text{ implying } f(y)\geq f(\xstar)
\end{align*}
Hence, we have $f(y)=f(\xstar)$.
\end{proof}

\begin{lemma}
$\abs{A}\leq 2C(f)$.
\end{lemma}

\begin{proof}
It suffices to show that $C(f)\geq C(g_b)$ for $b\in \B$. We will show that any certificate $B$ for $f$ on $x$ is also a certificate for $g_b$. For all $y,y'$ with $y\vert_C=y'\vert_C$, we have $f(y) = f(y')$, but this implies by definition that $g_b(y)=g_b(y')$. Hence, $B$ is also a certificate for $g_b$.
\end{proof}

Combining these lemmas concludes the proof of Theorem~\ref{thm:reals}.
\section{Finding the Shortest Certificate for a Monotone Function}

In this section, we show that there exists a family of instances on which the problem of finding the shortest certificate for a binary-valued $f$ on an input $\xstar$ (denoted $k\defeq C(f,\xstar)$) requires at least $\Omega\inparen{\binom{n}{k}}$ queries. 

Notice that this result is essentially optimal: for any function $f$, any input $\xstar$ and $k=C(f,\xstar)$, $O\inparen{\binom{n}{k}}$ suffice to find a size-$k$ certificate. Assuming $f(\xstar)=1$, the algorithm can simply query $f(x_S)$ for all subsets $S$ of size $k$ and check if each one of them is a certificate.

\begin{definition}[$F_k$] We define the set of $k$-indicator functions, denoted $F_k$ as follows.

Let $\myfunc{f_P}{\B^n}{\B}$ for some $P\subset[n]$ be defined as follows:
    \[
    f_P(x)=\begin{cases}
    0 & \abs{S_x}<k\\
    1 & \abs{S_x}>k\\
    0 & \abs{S_x}=k, P\neq S_x\\
    1 & \abs{S_x}=k, P=S_x
    \end{cases}
    \]
    Finally, let $F_k=\inbraces{f_P\suchthat \abs{P}=k}$.

\end{definition}

\begin{lemma}
\label{lemma:cert_complexity_fk}
Every function $f\in F_k$ has $C(f,\onev^n)=k$.
\end{lemma}
\begin{proof}
It is easy to see that every function in $F_k$ is monotone for all $k$.

Let $P \subset [n]$ be such that $f_P = f$. Observe that $\abs{P} = k$. Next, notice that $f(\onev^n_P) = f(\onev^n) = 1$. This implies that $C(f,\onev^n) \le \abs{P} = k$. Finally, consider any $S \subset [n]$ such that $\abs{S} < k$. Note that for $x = \onev^n_S$, we have $\abs{S_x} < k$, so $f_P(x) = 0$. Thus, we have $C(f,\onev^n) \ge k$, and we're done.
\end{proof}

\begin{theorem}
For any $k \in [n-1]$ and (randomized) algorithm $\cA$, there exists a function $\myfunc{f}{\B^n}{\B}$ and input $\xstar$ with $C(f,\xstar)=k$ such that $\cA$ must make at least $\nfrac{1}{2}\cdot\binom{n}{k}$ queries to $f$ to find the size-$k$ certificate with probability $\ge\nfrac{1}{2}$.
\end{theorem}

\begin{proof}
Fix an arbitrary $k \in [n-1]$. We will show that some function $f\in F_k$ takes $\ge\nfrac{1}{2}\cdot\binom{n}{k}$ queries to certify on the input $\onev^n$. Note that $C(f,\onev^n) = k$. Additionally, observe that any randomized algorithm to find a certificate for $f(\onev^n)=1$ can be converted to one that only makes queries $x$ satisfying $\abs{S_x} = k$. 

Let $X=\inbraces{x\in\B^n \suchthat \abs{S_x} = k}$. Notice that $\abs{X} = \binom{n}{k}$. Any randomized algorithm for finding the single $x \in X$ such that $f(x)=1$ can be viewed as one that samples a permutation from some distribution over permutations of $X$ and makes queries to $f$ in the order determined by the permutation until the algorithm encounters the $x\in X$ for which $f(x) = 1$. This is because query $i$ only depends on the values of the queries $1, \dots, i - 1$, and not their responses -- in particular, the responses to queries $1, \dots, i - 1$ are all $0$ if the algorithm has not terminated prior to issuing query $i$. With this interpretation in mind, fix some distribution of permutations of $X$; call this distribution $\cP$. 

For each $x\in X$, consider $\prvv{P\in \cP}{P^{-1}(x)\le\nfrac{\abs{X}}{2}}$ where $P^{-1}(x)$ is the index of element $X$. Let $\mu(P)$ denote the probability that a random permutation drawn from $\cP$ is $P$, and observe the following manipulations:
\begin{align*}
    \sum_{x\in X} \prvv{P\in \cP}{P^{-1}(x)\le\nfrac{\abs{X}}{2}} &= \sum_{x\in X} \sum_{P} \mu(P) \cdot \indicator{P^{-1}(x) \le \nfrac{\abs{X}}{2}}\\
    &= \sum_{P} \mu(P) \cdot \sum_{x \in X} \indicator{P^{-1}(x) \le \nfrac{\abs{X}}{2}}\\
    &= \sum_{P} \mu(P) \cdot \frac{1}{2}\cdot\abs{X} = \frac{1}{2}\cdot\abs{X}
\end{align*}
Thus, there exists at least one $x \in X$ for which $\prvv{P\in \cP}{P^{-1}(x)\le \nfrac{\abs{X}}{2}} \le \nfrac{1}{2}$. It follows that the algorithm does not find a sized-$k$ subset of $\onev^n$ certifying $f(\onev^n)=1$ with probability $>\frac12$ without making at least $\frac12\binom{n}{k}$ queries.
\end{proof}

\nocite{*}
\printbibliography

\end{document}